\documentclass[conference]{IEEEtran}
\usepackage{amsmath,amssymb,amsthm}
\usepackage{graphicx}
\graphicspath{{figures/}{./}}
\usepackage{booktabs}
\usepackage{braket}
\usepackage{tikz}
\usetikzlibrary{arrows.meta,positioning,fit,backgrounds}
\usepackage{quantikz}
\usepackage[hidelinks]{hyperref}

\newtheorem{proposition}{Proposition}
\newtheorem{theorem}{Theorem}
\newtheorem{lemma}{Lemma}
\newtheorem{definition}{Definition}

\DeclareMathOperator{\Tr}{Tr}
\DeclareMathOperator{\supp}{supp}
\DeclareMathOperator{\Comm}{Comm}

\IEEEoverridecommandlockouts

\begin{document}

\title{Blind Symmetry Matching in Quantum States,\\ with Application to Shot-Count Reduction}

\author{\IEEEauthorblockN{Mitchell A. Thornton\thanks{M. A. Thornton is with the
Darwin Deason Institute for Cyber Security, Department of Electrical and Computer
Engineering, Southern Methodist University, Dallas, Texas, USA (e-mail:
mitch@smu.edu).}}}

\maketitle

\begin{abstract}
Measuring a quantum computation in a basis adapted to a symmetry it carries reduces
the repeated measurements, commonly referred to as ``shots'' in the quantum computing
community, needed to read a statistical answer. Detecting the symmetry a quantum state carries has many uses, among them certifying a
claimed symmetry, identifying a conserved-charge sector, flagging symmetry-breaking as an
error signature, and selecting a compression or readout basis; shot-count reduction is
chosen here as one exemplary use case. Existing methods assume the symmetry is known in
advance; we remove that assumption. When the carried symmetry is unknown, it is
discovered from the data by a symmetry test that scores candidate groups, and the largest
passing group is exploited as the measurement basis. We state the pipeline precisely,
prove the selection rule is unbiased, and charge discovery in full. Two conditions are
treated, both detected by the same score with a different projection: a weak condition,
commutation with the representation, and a strong condition, confinement to a single
charge sector, the distinction drawn in the quantum-reference-frame literature. Seeded
demonstrations show the loop wins net of discovery: weak matching on momentum readout
reduces shots by a factor widening from ten to several thousand, and strong matching on a
two-system target reduces them by a further factor of the subsystem size. Blind symmetry
matching is a practical shot-reduction preprocessor.
\end{abstract}


\section{Introduction}

A quantum computation ends in a measurement, and the cost of that measurement, counted
in repeated state preparations, the number of repeated measurements commonly referred
to as ``shots'' in the quantum computing community, is often the binding resource on
near-term hardware. When the quantity of interest is statistical, an expectation value
or a distribution rather than a single sharp outcome, measuring in a basis adapted to a
symmetry the target carries reads that quantity in fewer shots than computational-basis
readout. This reduction is established and exploited across several settings: grouping
commuting Pauli terms in variational energy
estimation~\cite{verteletskyi2020measurement,yen2020measuring,crawford2021efficient,
gokhale2020measurement}, symmetry-adapted classical shadows~\cite{zhao2024group,
sauvage2024classical}, and symmetry-resolved estimation of structured
states~\cite{thornton2026quantumalgebraicdiversitysinglecopy}.

Every one of these methods assumes the symmetry is known in advance, read off the
structure of the problem. For a target whose structure is not given, the matched basis
cannot be written down, and the reduction is unavailable. A separate line of work tests
or identifies symmetry from data~\cite{laborde2023testing,laborde2022hamiltonian,
hayashi2025predicting}, but it stops at the test, certifying or rejecting a symmetry
without closing the loop to a reduced measurement cost on a target.

This report closes that loop and states the pipeline precisely. The contribution is a
procedure, demonstrated end to end, that discovers the matched symmetry group from the
data and then exploits it to reduce the shot count, and that wins net of its own
discovery cost. The discovery step scores a family of candidate groups by a symmetry
test, each score a single scalar estimable from few shots, and selects the largest group
whose score certifies a match. The exploitation step measures the target in the selected
symmetry-adapted basis. The discovery cost is charged in full; the central finding is
that it is cheap relative to the savings and grows cheaper as the target grows larger.

Blind symmetry matching is a primitive, not a single application. Once the symmetry an
unknown state carries is discovered, the same information serves many ends. It certifies a
claimed symmetry, the task of symmetry testing, by reporting whether a candidate group is
carried. It identifies which conserved-charge or superselection sector a state occupies,
the strong mode below. It flags symmetry-breaking as a signature of noise or error, since a
state that should be symmetric but scores below one has been corrupted. It selects a
symmetry-adapted basis for compression, for state verification, or for readout. And it
prepares the ground for symmetry-resolved estimation of any functional of the state. This
report develops one of these, shot-count reduction, in full and to a demonstrated win,
because it admits a clean end-to-end accounting; the discovery primitive and the single
circuit that realizes it (Section~\ref{sec:circuit}) are common to all of them, and the
generality is not confined to the cyclic groups of the demonstrations
(Section~\ref{sec:nonabelian}). The code reproducing every result accompanies this
paper.\footnote{\url{https://github.com/mitch-thornton/quantum-state-symmetry-discovery}}

We treat two notions of what it means for a target to carry a symmetry, following a
distinction made precise in the quantum-reference-frame
literature~\cite{doat2026symmetry}. The \emph{weak} condition is commutation with the
group representation, equivalently block-diagonality in the charge basis; its projector
is the incoherent group average, the twirl onto the commutant, which is the object
underlying algebraic diversity~\cite{thornton2026algebraicdiversitygrouptheoreticspectral,
thornton2026algebraicdiversityprinciplesgrouptheoretic} and geometric quantum machine
learning~\cite{ragone2022representation,larocca2022groupinvariant}. The \emph{strong}
condition is confinement to a single charge sector; its projector is the coherent group
average. Both conditions are detected by the same normalized score with a different
projection, and both drive the same pipeline. The strong condition is more restrictive
and, when it holds, yields a larger reduction.

\section{Symmetry-adapted measurement and the matching score}
\label{sec:score}

Let a finite group $G$ act on $\mathbb{C}^d$ by a unitary representation $\pi$, with
$U_g=\pi(g)$. Two projections express the two notions of symmetry. The \emph{weak twirl}
\begin{equation}
\mathcal{T}^G_W(R)=\frac{1}{|G|}\sum_{g\in G}U_g\,R\,U_g^\dagger
\label{eq:weak}
\end{equation}
is the orthogonal projection onto the commutant
$\Comm(\pi)=\{X:XU_g=U_gX\;\forall g\}$, equivalently the block-diagonal operators in the
charge basis. The \emph{strong twirl}
\begin{equation}
\mathcal{T}^G_S(R)=\frac{1}{|G|^2}\sum_{g,g'\in G}U_g\,R\,U_{g'}^\dagger
=P_0\,R\,P_0
\label{eq:strong}
\end{equation}
factors through $P_0=\frac{1}{|G|}\sum_g U_g$, the projector onto the invariant
subspace $H_0=\{\ket{\psi}:U_g\ket{\psi}=\ket{\psi}\;\forall g\}$, so
$\mathcal{T}^G_S$ is compression to the zero-charge sector. Replacing $P_0$ by the
projector $P_c$ onto any fixed charge sector gives the same construction for confinement
to sector $c$, the fixed-particle-number situation. Both maps are orthogonal projections
in the Hilbert-Schmidt inner product $\langle X,Y\rangle=\Tr[X^\dagger Y]$.

\begin{definition}[matching score]
For an orthogonal projection $\mathcal{T}$ on operators, the matching score of a target
$R\neq 0$ is
\begin{equation}
s_{\mathcal{T}}(R)=\frac{\langle \mathcal{T}(R),R\rangle}{\langle R,R\rangle}
=\frac{\Tr[\mathcal{T}(R)\,R]}{\Tr[R^2]}.
\label{eq:score}
\end{equation}
\end{definition}

\begin{proposition}[go/no-go]
\label{prop:gonogo}
$s_{\mathcal{T}}(R)\in[0,1]$, and $s_{\mathcal{T}}(R)=1$ if and only if
$\mathcal{T}(R)=R$, that is if and only if $R$ carries the symmetry tested by
$\mathcal{T}$.
\end{proposition}
\begin{proof}
Since $\mathcal{T}$ is an orthogonal projection, $R=\mathcal{T}(R)+R^\perp$ with
$\mathcal{T}(R)\perp R^\perp$, so $\langle\mathcal{T}(R),R\rangle=\|\mathcal{T}(R)\|^2$
and $s_{\mathcal{T}}(R)=\|\mathcal{T}(R)\|^2/\|R\|^2$. A projection is contractive, so
$0\le\|\mathcal{T}(R)\|^2\le\|R\|^2$, with the upper equality exactly when
$R^\perp=0$, i.e. $\mathcal{T}(R)=R$.
\end{proof}

Write $s^G_W$ and $s^G_S$ for the scores built from $\mathcal{T}^G_W$ and
$\mathcal{T}^G_S$. The acceptance probability of the symmetry tests
of~\cite{laborde2023testing,laborde2022hamiltonian} realizes $s^G_W$ on a quantum
computer without reconstructing $R$; $s^G_S$ is realized by projecting onto the sector,
with acceptance probability $\Tr[R\,P_0]$. Identifying a symmetry from data has also
been framed as subgroup hypothesis testing, for predicting symmetries of unknown
dynamics rather than selecting a group to reduce shots~\cite{hayashi2025predicting}.

The selection rule rests on two containment facts.

\begin{proposition}[under-selection is safe]
\label{prop:subgroup}
If $H\le G$ then $s^G_W(R)=1\Rightarrow s^H_W(R)=1$ and
$s^G_S(R)=1\Rightarrow s^H_S(R)=1$.
\end{proposition}
\begin{proof}
A larger group imposes more commutation constraints, so
$\Comm(\pi|_G)\subseteq\Comm(\pi|_H)$ and the range of $\mathcal{T}^G_W$ is contained in
that of $\mathcal{T}^H_W$; the weak claim follows from Proposition~\ref{prop:gonogo}. For
the strong case $H_0^G=\bigcap_{g\in G}\ker(U_g-I)\subseteq\bigcap_{h\in
H}\ker(U_h-I)=H_0^H$, so the range of $\mathcal{T}^G_S$ is contained in that of
$\mathcal{T}^H_S$.
\end{proof}

\begin{proposition}[strong implies weak]
\label{prop:nest}
$s^G_S(R)=1\Rightarrow s^G_W(R)=1$.
\end{proposition}
\begin{proof}
If $\supp(R)\subseteq H_0$ then, since every $U_g$ acts as the identity on $H_0$,
$U_g R U_g^\dagger=R$ for all $g$, hence $\mathcal{T}^G_W(R)=R$ and
$s^G_W(R)=1$ by Proposition~\ref{prop:gonogo}.
\end{proof}

Conversely, over-selecting a group the target does not carry scores below $1$: if $R$
carries exactly $G_\star$ and $G\supsetneq G_\star$ acts so that
$\mathcal{T}^G(R)\neq R$, then $s^G(R)<1$ and $G$ is rejected. The selection rule is
therefore to take the largest group whose score certifies a match, which by
Proposition~\ref{prop:subgroup} never imposes a symmetry the target lacks.

The cost of estimating a score is set by the gap that separates a match from the nearest
over-selection. Distinguishing $s=1$ from $s=1-\gamma$ to confidence requires resolving
$\gamma$, so a candidate is scored in $O(\gamma^{-2})$ shots, improved to
$O(\gamma^{-1})$ by amplitude estimation~\cite{brassard2002amplitude}.

\section{The blind-matching pipeline}
\label{sec:pipeline}

\subsection{When shot-count reduction is possible}
\label{sec:screen}

Shot-count reduction is the application developed here, and not every computation admits it.
A screen applied before the pipeline decides eligibility.

\begin{lemma}[no-go for sharp targets]
\label{lem:nogo}
If the target is a single deterministic outcome carried by a state
$\rho=\ket{b^\ast}\!\bra{b^\ast}$ sharp in a basis $B^\ast$, then the shot cost is $O(1)$ and
no change of basis lowers it; any basis mutually unbiased to $B^\ast$ returns an outcome
distribution independent of $b^\ast$, carrying no information about it.
\end{lemma}

A sharp answer is read in one shot in its own basis and erased in the conjugate basis,
leaving no graded quantity for a symmetry-adapted basis to improve. A reduction therefore
requires a \emph{statistical} target, an expectation value or a distribution. The eligible
case is the converse: a statistical target whose state carries a nontrivial symmetry whose
adapted basis block-diagonalizes the estimation, with the reduction exceeding the discovery
overhead. Large reductions come from collapsing measurement settings, smaller ones from
reducing the estimation dimension.

This screen gates the shot-reduction use case only. The other uses of blind symmetry
matching, certifying a claimed symmetry, identifying the occupied charge sector, and flagging
symmetry-breaking as an error signature, do not require a statistical target and are not
subject to it: a sharp state has a perfectly well-defined symmetry to discover. The screen is
a precondition for this application, not for the primitive.

A computation that returns a sharp answer can sometimes be modified to return a statistical
one, moving it from the no-go case into eligibility. Textbook Grover search marks a single
basis state and amplifies it to a sharp outcome, the no-go case of Lemma~\ref{lem:nogo}.
Replacing the single-state oracle with one that marks a \emph{region}, a projector-valued
oracle $I-2P_\lambda$ that amplifies an entire subspace rather than one
state~\cite{brassard2002amplitude}, makes the output a distribution over that region, a
statistical target eligible for reduction. The region is not arbitrary: to carry a usable
symmetry its marked entries must form a group orbit or an isotypic sector, so that the
amplified state lands in a symmetry-adapted subspace and the discovered group is the one that
block-diagonalizes it. Marking a sector this way reaches the strong condition by amplitude
amplification rather than by post-selection.

The target is a statistical quantity of an unknown state $R$ whose carried symmetry is
not given. A structure-blind readout estimates the quantity without using any symmetry,
at a shot cost $N_{\mathrm{blind}}$ that grows with the dimension of the structure-blind
parameterization. The pipeline replaces this with the procedure of
Fig.~\ref{fig:pipeline}, stated precisely as follows and realized as a single quantum
circuit in Section~\ref{sec:circuit}.

\begin{figure}[t]
\centering
\resizebox{\columnwidth}{!}{%
\begin{tikzpicture}[
  font=\footnotesize,
  box/.style={draw, rounded corners, align=center, minimum height=7mm, inner sep=3pt},
  >={Stealth[length=2mm]}]
\node[box] (t) {black-box\\ target $R$};
\node[box, right=5mm of t, fill=blue!6] (d) {\textbf{discovery}\\ score
  $\{G_i\}$ by\\ symmetry test};
\node[box, right=5mm of d, fill=blue!6] (s) {\textbf{select}\\ largest\\ passing $G^\star$};
\node[box, right=5mm of s, fill=green!7] (e) {\textbf{exploit}\\ measure in\\ adapted basis};
\node[box, right=5mm of e] (o) {estimate};
\draw[->] (t) -- (d);
\draw[->] (d) -- (s);
\draw[->] (s) -- (e);
\draw[->] (e) -- (o);
\node[below=5.5mm of d, align=center] (cd) {$N_{\mathrm{disc}}=m\,O(\gamma^{-2})$};
\node[below=5.5mm of e, align=center] (ce) {$N_{\mathrm{exploit}}$};
\draw[->, dashed] (d) -- (cd);
\draw[->, dashed] (e) -- (ce);
\node[align=center, below=2mm of cd.south -| s] (win)
  {win: $N_{\mathrm{disc}}+N_{\mathrm{exploit}}<N_{\mathrm{blind}}$};
\end{tikzpicture}}
\caption{The blind-matching pipeline. Discovery scores a candidate family by a symmetry
test (weak twirl or strong sector projection), the selection rule takes the largest
passing group, and exploitation reads the target in its symmetry-adapted basis. Here $R$ is the
unknown state, or its covariance, whose carried symmetry is to be discovered. Every shot,
discovery included, is charged against the structure-blind baseline.}
\label{fig:pipeline}
\end{figure}
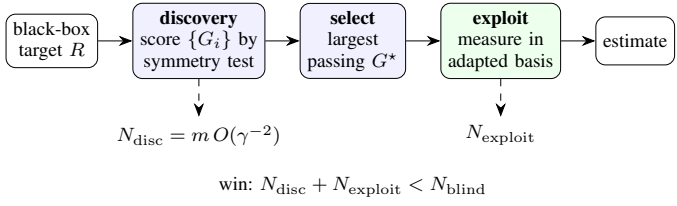

\noindent\textbf{Pipeline.}
\emph{Inputs:} black-box preparations of $R$; a candidate family
$\{G_1,\dots,G_m\}$ ordered so that a larger group, when carried, yields a larger
reduction; a mode, weak ($\mathcal{T}_W$) or strong ($\mathcal{T}_S$); a target accuracy
$\varepsilon$; a confidence level.
\begin{enumerate}
\item \emph{Discovery.} For each $G_i$, estimate $s_{\mathcal{T}}^{G_i}(R)$ from
$N_{\mathrm{score}}$ symmetry-test shots. In the weak mode the test is the twirl
acceptance of~\cite{laborde2023testing}; in the strong mode it is projection onto the
candidate sector, and once a group is fixed the occupied sector is read directly from the
charge observable.
\item \emph{Selection.} Accept candidates whose estimated score exceeds a threshold
$\tau$ placed between $1$ and the best rejected score, and set $G^\star$ to the largest
accepted group. If none pass, fall back to the structure-blind readout.
\item \emph{Exploitation.} Measure $R$ in the symmetry-adapted basis of $G^\star$: the
block-diagonal basis in the weak mode, the matched sector in the strong mode, estimating
the target functional to accuracy $\varepsilon$ at cost $N_{\mathrm{exploit}}$.
\end{enumerate}
\emph{Output:} an estimate of the target, at total cost
$N_{\mathrm{disc}}+N_{\mathrm{exploit}}$ with
$N_{\mathrm{disc}}=m\,N_{\mathrm{score}}$.

\begin{theorem}[win condition]
\label{thm:win}
The pipeline reads the target in fewer shots than the structure-blind baseline whenever
\begin{equation}
m\,N_{\mathrm{score}}+N_{\mathrm{exploit}}<N_{\mathrm{blind}}.
\label{eq:win}
\end{equation}
Since $N_{\mathrm{score}}=O(\gamma^{-2})$ depends on the candidate count and the score
gap while $N_{\mathrm{blind}}/N_{\mathrm{exploit}}$ is the reduction factor of the matched
basis, the inequality holds with widening margin once the reduction factor grows past the
fixed discovery overhead.
\end{theorem}

The reduction factor is mode-dependent. In the weak mode the matched basis is
block-diagonal, so $N_{\mathrm{exploit}}$ scales with $\sum_l d_l^2$ against
$N_{\mathrm{blind}}\sim d^2$ for $d=\sum_l d_l$. In the strong mode the target is
confined to one sector of dimension $d_c$, so $N_{\mathrm{exploit}}\sim d_c^2$ against
$d^2$, a larger factor $(d/d_c)^2$. The two demonstrations below instantiate each mode.

\section{Circuit realization}
\label{sec:circuit}

The discovery score~\eqref{eq:score} is the acceptance statistic of a single quantum
circuit, a controlled twirl followed by a Hilbert-Schmidt test in the symmetry-testing
family of~\cite{laborde2023testing,laborde2022hamiltonian}. The same circuit
(Fig.~\ref{fig:circuit}) certifies both conditions, and the only difference between them
is the fate of the group register. We give the construction in full so that the pipeline
is implementable from this paper alone.

The target enters as the density matrix $\rho=R/\Tr R$; the score is scale-invariant, so
$s_{\mathcal{T}}(\rho)=s_{\mathcal{T}}(R)$. A control register $C$ of
$\lceil\log_2|G|\rceil$ qubits, two copies $S_1,S_2$ of $\rho$, and one test ancilla $a$
are prepared, with $C$ in the uniform group state
$\ket{+_G}=|G|^{-1/2}\sum_{g\in G}\ket{g}$. The controlled group action
\begin{equation}
\mathrm{C}U=\sum_{g\in G}\ket{g}\!\bra{g}_C\otimes (U_g)_{S_1}
\label{eq:cu}
\end{equation}
acts on the control and the first copy, leaving the $C$--$S_1$ state
\begin{equation}
\frac{1}{|G|}\sum_{g,g'\in G}\ket{g}\!\bra{g'}_C\otimes U_g\,\rho\,U_{g'}^\dagger .
\label{eq:joint}
\end{equation}
The two readouts of the control register project~\eqref{eq:joint} onto the two notions of
symmetry. Discarding $C$ leaves $S_1$ in
\begin{equation}
\Tr_C\!\left[\frac{1}{|G|}\sum_{g,g'}\ket{g}\!\bra{g'}\otimes U_g\rho U_{g'}^\dagger\right]
=\frac{1}{|G|}\sum_{g}U_g\,\rho\,U_g^\dagger=\mathcal{T}^G_W(\rho),
\label{eq:weakcirc}
\end{equation}
the incoherent group average, the Reynolds projection onto the commutant. Post-selecting
$C$ on $\ket{+_G}$ instead leaves $S_1$ in
\begin{equation}
\frac{(\bra{+_G}\!\otimes\! I)\,[\,\cdot\,]\,(\ket{+_G}\!\otimes\! I)}{\Tr[P_0\rho P_0]}
=\frac{\mathcal{T}^G_S(\rho)}{\Tr[P_0\rho P_0]},
\label{eq:strongcirc}
\end{equation}
the coherent group average, heralded with probability $\Tr[P_0\rho P_0]$, the weight of
$\rho$ in the sector. Equations~\eqref{eq:weakcirc} and~\eqref{eq:strongcirc} are the
circuit-level form of the incoherent and coherent group averaging
of~\cite{doat2026symmetry}: discarding the group register tests the weak condition,
post-selecting it tests the strong condition.

A SWAP test between $S_1$, now carrying $\mathcal{T}(\rho)$, and the untwirled copy $S_2$
accepts with probability
\begin{equation}
P_{\mathrm{acc}}=\tfrac12\bigl(1+\Tr[\mathcal{T}(\rho)\,\rho]\bigr),
\label{eq:swap}
\end{equation}
and a SWAP test on two untwirled copies returns
$P_{\mathrm{pur}}=\tfrac12(1+\Tr[\rho^2])$. The matching score follows directly:
\begin{equation}
s^G_W=\frac{2P_{\mathrm{acc}}-1}{2P_{\mathrm{pur}}-1},\qquad
s^G_S=\Tr[P_0\rho P_0]\,\frac{2P_{\mathrm{acc}}-1}{2P_{\mathrm{pur}}-1},
\label{eq:scorecirc}
\end{equation}
with $\mathcal{T}=\mathcal{T}^G_W$ or $\mathcal{T}^G_S$ in~\eqref{eq:swap}. Amplitude
estimation~\cite{brassard2002amplitude} on the acceptance amplitude reads
$P_{\mathrm{acc}}$ to accuracy $\varepsilon$ in $O(\varepsilon^{-1})$ rather than
$O(\varepsilon^{-2})$ uses of the circuit, the quadratic acceleration quantified in
Section~\ref{sec:weak}. A linear-algebra check that the circuit
scores~\eqref{eq:scorecirc} reproduce the defining scores~\eqref{eq:score} is included in
the bundle.

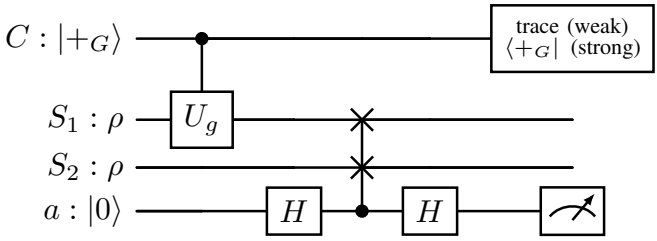
\begin{figure}[t]
\centering
\resizebox{\columnwidth}{!}{%
\begin{quantikz}[row sep=0.22cm, column sep=0.38cm]
\lstick{$C:\ket{+_G}$} & \ctrl{1} & \qw & \qw & \qw & \gate{\substack{\text{trace (weak)}\\[-1pt]\bra{+_G}\ \text{(strong)}}} \\
\lstick{$S_1:\rho$} & \gate{U_g} & \qw & \swap{1} & \qw & \qw \\
\lstick{$S_2:\rho$} & \qw & \qw & \targX{} & \qw & \qw \\
\lstick{$a:\ket{0}$} & \qw & \gate{H} & \ctrl{-1} & \gate{H} & \meter{}
\end{quantikz}}
\caption{The discovery circuit. A controlled group action twirls one copy of the target;
discarding the group register $C$ realizes the weak twirl (incoherent average), while
post-selecting it on $\ket{+_G}$ realizes the strong twirl (coherent average, heralded by
the sector weight). The SWAP test on ancilla $a$ returns the matching score. The
controlled twirl and Hilbert-Schmidt test are in the symmetry-testing family
of~\cite{laborde2023testing}; amplitude estimation accelerates the readout. For a cyclic
group $\ket{+_G}=H^{\otimes n}\ket{0}$ and the controlled action is one modular adder.}
\label{fig:circuit}
\end{figure}

For a cyclic group the circuit is elementary. The uniform state is
$\ket{+_G}=H^{\otimes n}\ket{0}$; the controlled group action~\eqref{eq:cu} is a single
in-place modular adder, since $U_g$ translates by $g$; and exploitation applies the
quantum Fourier transform to reach the momentum, or charge, basis, after which a
computational-basis measurement reads the matched diagonal directly. In the strong mode
the charge labeling the sector is, for the two-system target, the sum of the two Fourier
indices modulo $L$, recovered from the same measurement. The basis-change unitary is the
only target-dependent part of exploitation, fixed once $G^\star$ is selected.

\section{Weak-symmetry demonstration}
\label{sec:weak}

We instantiate the weak mode on momentum readout of a translation-invariant chain, the
cyclic case where the symmetry-adapted basis is the Fourier basis and the structure-blind
baseline is correlation-matrix tomography. The target is a circulant covariance $R$ on
$L$ sites with an exponentially decaying correlation profile, carrying $\mathbb{Z}_L$.
The estimator is not told this. The quantity to estimate is the momentum distribution,
the Fourier-diagonal of $R$, to a fixed root-mean-square accuracy. All results are seeded
and reproduced by the bundle.

The candidate family is $\{e\}$, $\mathbb{Z}_L$, and the over-selection $S_L$.
Table~\ref{tab:scores} reports the exact scores. The trivial group and $\mathbb{Z}_L$
score $1$; the rule selects the larger, $\mathbb{Z}_L$. The over-selection $S_L$ scores
below $1$ with a gap that widens with $L$, from $0.10$ at $L=8$ to $0.50$ at $L=128$, so
discovery grows cheaper as the chain lengthens. Scoring each candidate from a finite
number of symmetry-test shots and applying the selection rule recovers $\mathbb{Z}_L$
from $50$ shots per candidate in $98.7\%$ of seeded trials and from $100$ shots in
$100\%$ (Fig.~\ref{fig:weakloop}A).

\begin{table}[t]
\centering
\caption{Weak-mode scores~\eqref{eq:score}. The trivial group and the true
$\mathbb{Z}_L$ certify a match; $S_L$ is rejected, with a gap that widens with $L$.}
\label{tab:scores}
\begin{tabular}{rcccc}
\toprule
$L$ & $s_{\{e\}}$ & $s_{\mathbb{Z}_L}$ & $s_{S_L}$ & gap $1-s_{S_L}$\\
\midrule
8   & 1.000 & 1.000 & 0.901 & 0.099\\
32  & 1.000 & 1.000 & 0.604 & 0.396\\
128 & 1.000 & 1.000 & 0.497 & 0.503\\
\bottomrule
\end{tabular}
\end{table}

Charging discovery in full, Table~\ref{tab:weakloop} compares the total pipeline cost
against structure-blind tomography at accuracy $0.02$. The pipeline wins at every size,
and the margin widens with $L$: the structure-blind cost grows as $L^2$ while the total
stays nearly flat, since discovery cost falls and the matched exploitation is a single
measurement setting. The advantage runs from about thirteenfold at $L=8$ to several
thousandfold at $L=128$ (Fig.~\ref{fig:weakloop}B). The exploitation reduction, the
$L^2/2$ factor of momentum readout over correlation tomography, is validated separately
in the bundle. The scoring step inherits the quadratic advantage of the symmetry test: on
a covariance carrying $D_4$, the matched group and its subgroups score exactly $1$ while
a mismatch and an over-selection score $0.978$, and an amplitude-estimation readout of
the score converges at rate $M^{-0.99}$ against $M^{-0.50}$ for direct
sampling~\cite{brassard2002amplitude}.

\begin{table}[t]
\centering
\caption{Weak-mode closed-loop budget at accuracy $0.02$ with three candidates. Total
(discovery plus exploitation) against structure-blind tomography.}
\label{tab:weakloop}
\begin{tabular}{rccccc}
\toprule
$L$ & $N_{\mathrm{disc}}$ & $N_{\mathrm{exploit}}$ & $N_{\mathrm{blind}}$ &
$N_{\mathrm{total}}$ & ratio\\
\midrule
8   & 2775 & 1837 & 58{,}757     & 4612 & 12.7$\times$\\
32  & 174  & 2331 & 1{,}193{,}442 & 2505 & 476.4$\times$\\
128 & 108  & 2458 & 20{,}133{,}768 & 2566 & 7846.4$\times$\\
\bottomrule
\end{tabular}
\end{table}

\begin{figure}[t]
\centering
\includegraphics[width=\columnwidth]{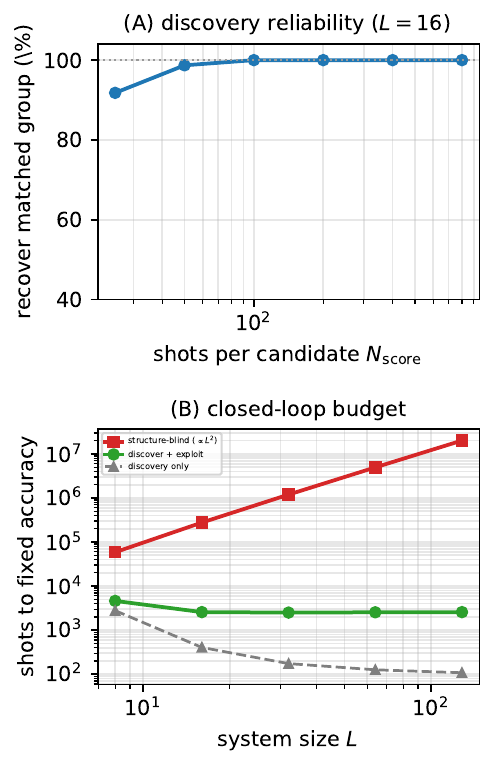}
\caption{Weak-mode loop on momentum readout. (A) Discovery recovers the matched group
from a few symmetry-test shots per candidate. (B) The structure-blind cost grows as
$L^2$ while the total pipeline cost stays nearly flat; discovery cost alone falls with
$L$.}
\label{fig:weakloop}
\end{figure}

\section{Strong-symmetry discovery}
\label{sec:strong}

The strong twirl that defines this mode is the coherent group average
\begin{equation}
\mathcal{T}^G_S(R)=\frac{1}{|G|^2}\sum_{g,g'\in G}U_g\,R\,U_{g'}^\dagger=P_0\,R\,P_0,
\label{eq:strongVI}
\end{equation}
set against the weak twirl $\mathcal{T}^G_W(R)=\frac{1}{|G|}\sum_g U_g R U_g^\dagger$
of~\eqref{eq:weak}, and the contrast between the two is the source of the names. The weak
twirl averages over a single copy of the group, normalization $1/|G|$, and projects onto the
commutant. The strong twirl averages over an independent pair $(g,g')$, normalization
$1/|G|^2$, the double sum factoring through the invariant-sector projector
$P_0=\frac{1}{|G|}\sum_g U_g$ as $P_0 R P_0$. The squared group order in the denominator marks
the coherent average, in which the two sums interfere, against the incoherent average of the
weak twirl, in which a single sum merely mixes; the confinement the coherent average enforces
is the stronger of the two conditions.

The strong condition, confinement to a single charge sector, is more restrictive than
commutation and, when it holds, yields a larger reduction: the target lives in a subspace
of dimension $d_c\ll d$, so estimating it there costs $d_c^2$ parameters against $d^2$.
This is the situation of a state in a fixed particle-number or fixed-charge sector, common
in chemistry and condensed matter, where the sector is usually known from the physics; the
pipeline supplies it when it is not.

We instantiate the strong mode on the two-system setting of~\cite{doat2026symmetry}: two
systems on an $L$-site ring, $G=\mathbb{Z}_L$ acting by the global shift $U\otimes U$, so
$H=\mathbb{C}^L\otimes\mathbb{C}^L$ carries $L$ charge sectors each of dimension $L$. A
single-register cyclic chain has a one-dimensional invariant subspace and so admits no
nontrivial strong reduction; the multiplicity of the two-system setting is what makes the
strong mode informative. The full target has $d^2=L^4$ parameters; a weakly symmetric
(block-diagonal) target has $L\cdot L^2=L^3$; a strongly symmetric (single-sector) target
has $L^2$. The reduction factors are $L$ for weak and $L^2$ for strong.

Table~\ref{tab:strongscores} reports the exact scores and confirms
Propositions~\ref{prop:gonogo} and~\ref{prop:nest} numerically. A strongly symmetric
target scores $1$ on both tests, consistent with strong implying weak. A block-diagonal
target that is not sector-confined scores $1$ on the weak test but strictly below on the
strong test, the deficit growing with $L$ as the mass spreads over more sectors, so the
strong score detects confinement that the weak score cannot. A generic target fails both.

\begin{table}[t]
\centering
\caption{Strong-mode scores in the two-system setting. A strongly symmetric target passes
both tests (Prop.~\ref{prop:nest}); a block-diagonal target passes weak but fails strong;
a generic target fails both.}
\label{tab:strongscores}
\begin{tabular}{rlcc}
\toprule
$L$ & target & weak score & strong score\\
\midrule
3 & strongly symmetric & 1.000 & 1.000\\
3 & weakly symmetric    & 1.000 & 0.428\\
3 & generic             & 0.675 & 0.255\\
5 & strongly symmetric & 1.000 & 1.000\\
5 & weakly symmetric    & 1.000 & 0.239\\
5 & generic             & 0.598 & 0.148\\
\bottomrule
\end{tabular}
\end{table}

The pipeline changes only in two places. Discovery is two-stage: the weak test first
identifies the group, after which the occupied sector is read directly from the charge
observable of that group, a single observable whose outcome labels the sector. Because
the charge measurement separates the true sector from the others by the full confinement
gap, sector discovery is nearly free: identifying the occupied sector at $L=16$ with a
$10\%$ leakage takes $5$ charge-measurement shots for $99.9\%$ reliability and $10$ shots
for $100\%$ (Fig.~\ref{fig:strongloop}A). Exploitation restricts tomography to the matched
sector.

Table~\ref{tab:strongloop} reports the closed-loop budget. The strong mode wins at every
size and the margin scales as $L^2$, an extra factor of $L$ over the weak mode in the same
setting, from $15.8\times$ at $L=4$ to $4095.8\times$ at $L=64$
(Fig.~\ref{fig:strongloop}B). The discovery overhead, dominated by the group-discovery
cost of the weak stage, is a fixed additive term that the $L^2$ reduction overtakes
immediately.

\begin{table}[t]
\centering
\caption{Strong-mode closed-loop budget (parameter-count proxy at accuracy $0.02$). The
weak reduction is $L$; the strong reduction is $L^2$.}
\label{tab:strongloop}
\begin{tabular}{rcccc}
\toprule
$L$ & $N_{\mathrm{disc}}$ & $N_{\mathrm{strong}}$ & weak ratio & strong ratio\\
\midrule
4  & 612 & 40{,}000     & 4.0$\times$  & 15.8$\times$\\
16 & 612 & 640{,}000    & 16.0$\times$ & 255.8$\times$\\
64 & 612 & 10{,}240{,}000 & 64.0$\times$ & 4095.8$\times$\\
\bottomrule
\end{tabular}
\end{table}

\begin{figure}[t]
\centering
\includegraphics[width=\columnwidth]{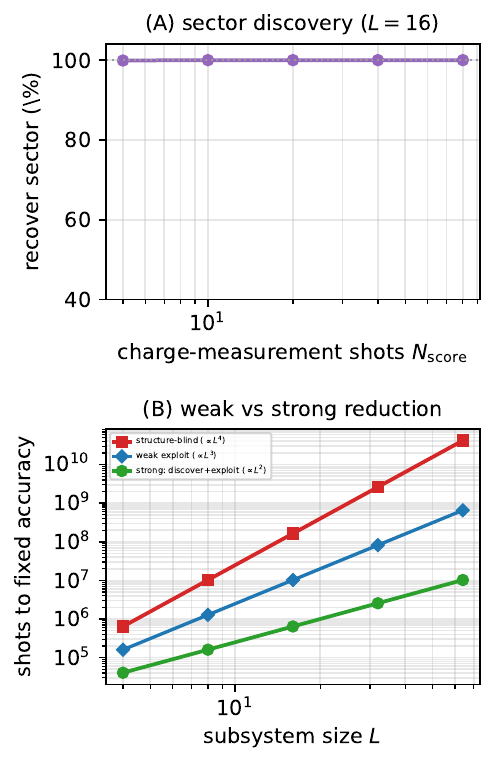}
\caption{Strong-mode loop on the two-system target. (A) The occupied sector is recovered
from a handful of charge-measurement shots. (B) Structure-blind cost grows as $L^4$, weak
exploitation as $L^3$, and strong discovery-plus-exploitation as $L^2$; the strong mode
buys an extra factor of $L$ over the weak mode.}
\label{fig:strongloop}
\end{figure}

\section{Beyond the cyclic case: non-Abelian symmetry}
\label{sec:nonabelian}

The twirl, the score, and the discovery circuit are defined for any finite group and any
unitary representation~\cite{serre1977linear,fulton1991representation}; the cyclic
demonstrations are the worked example, in which the symmetry-adapted basis is the Fourier
basis. For a general group the adapted basis is the group-Fourier, or isotypic,
decomposition, of which the Fourier basis is the cyclic instance, and the twirl is the
Reynolds projection onto the commutant in every case~\cite{werner1989quantum,bennett1996mixed}.
Two non-Abelian groups confirm that the go/no-go and nesting relations of
Section~\ref{sec:score} carry over unchanged.

The dihedral group $D_4$, in a two-copy action $U_g\otimes U_g$ on
$\mathbb{C}^4\otimes\mathbb{C}^4$ with a three-dimensional invariant sector, behaves exactly
as the cyclic case: a sector-confined target scores weak $1.000$ and strong $1.000$; a
block-diagonal target scores weak $1.000$ but strong $0.154$; a generic target fails both
($0.587$, $0.090$). The adapted basis here is the $D_4$ isotypic decomposition, not the
Fourier basis.

The symmetric group $S_n$, acting by permutation of $n$ qubits, is the more striking case.
Its adapted basis is the Schur-Weyl decomposition~\cite{bacon2006efficient,harrow2013church},
and the strong sector of greatest interest is the fully symmetric, or Dicke,
subspace~\cite{dicke1954coherence}, the invariant sector, of dimension $n+1$ against the
ambient $2^n$. Confinement to it is an exponential reduction.
Table~\ref{tab:schur} confirms the scores and reports the reduction: a Dicke-subspace target
scores one on both tests at every $n$; a permutation-invariant but unconfined target passes
weak and fails strong, with a gap that grows; and the strong reduction factor
$(2^n/(n+1))^2$ runs from $4$ at $n=3$ to $28$ at $n=5$ and exponentially beyond.

\begin{table}[t]
\centering
\caption{Symmetric group $S_n$: strong confinement to the symmetric (Dicke) subspace is an
exponential reduction. Scores are (weak / strong).}
\label{tab:schur}
\begin{tabular}{rcccc}
\toprule
$n$ & $2^n$ & Dicke dim $n{+}1$ & Dicke target & perm-invariant\\
\midrule
3 & 8  & 4 & 1.00 / 1.00 & 1.00 / 0.50\\
4 & 16 & 5 & 1.00 / 1.00 & 1.00 / 0.28\\
5 & 32 & 6 & 1.00 / 1.00 & 1.00 / 0.28\\
\bottomrule
\end{tabular}
\end{table}

Only the group-specific blocks of the discovery circuit change. The controlled group action
becomes a controlled permutation, a network of transpositions, the controlled SWAPs selected
by the control register, in place of the cyclic modular adder; the exploitation basis change
becomes the Schur transform~\cite{bacon2006efficient}, the $S_n$ analog of the quantum
Fourier transform; and the strong post-selection projects onto the symmetric subspace.
Figure~\ref{fig:schur} shows the modification. Discarding the control still gives the weak
test and post-selecting it the strong test, so the single circuit of
Section~\ref{sec:circuit} discovers both notions of symmetry for the non-Abelian case as
well.

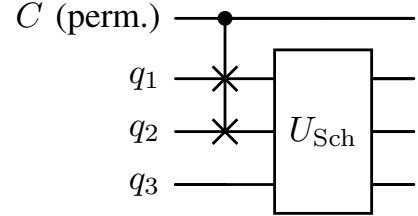
\begin{figure}[t]
\centering
\resizebox{0.63\columnwidth}{!}{%
\begin{quantikz}[row sep=0.24cm, column sep=0.42cm]
\lstick{$C$ (perm.)} & \ctrl{1} & \qw & \qw \\
\lstick{$q_1$} & \swap{1} & \gate[3]{U_{\mathrm{Sch}}} & \qw \\
\lstick{$q_2$} & \targX{} & & \qw \\
\lstick{$q_3$} & \qw & & \qw
\end{quantikz}}
\caption{Schur-Weyl instantiation of the discovery circuit of Fig.~\ref{fig:circuit} for the
symmetric group. The controlled group action becomes a controlled permutation, transpositions
realized as controlled SWAPs selected by the control register, replacing the cyclic modular
adder; the exploitation basis change becomes the Schur transform $U_{\mathrm{Sch}}$, the
$S_n$ analog of the QFT, after~\cite{bacon2006efficient}. Discarding the control register
gives the weak test; post-selecting it gives the strong test, which here projects onto the
symmetric (Dicke) subspace.}
\label{fig:schur}
\end{figure}

\section{The exploitation step}
\label{sec:exploit}

Once $G^\star$ and the mode are fixed, the exploitation step is the established
symmetry-adapted reduction, and the pipeline inherits whatever reduction that step
provides. In the weak mode the matched basis block-diagonalizes the target: the Fourier
basis removes the $L^2/2$ off-diagonal correlators of momentum readout, the Schur basis
reduces the copies of spectrum estimation~\cite{odonnell2016efficient,keyl2001estimating,
bacon2006efficient}, and commuting-term grouping reduces the number of settings in
variational energy estimation~\cite{verteletskyi2020measurement,yen2020measuring,
crawford2021efficient}. In the strong mode the matched sector confines the target to a
subspace, and estimation runs within it. The pipeline does not change these reductions; it
supplies the group, and the sector, when they are not given.

Lest the non-Abelian examples invite the comparison, the discovery step is not the
hidden-subgroup problem~\cite{simon1997power,childs2010quantum}: there the group is known
and the difficulty is extracting a hidden subgroup from coset states by measurements
entangled across many copies, hard for non-Abelian groups, whereas blind matching scores
which candidate group's representation commutes with a given state, or into which sector it
is confined, by the single-copy commutant or projection computation of~\eqref{eq:score},
with no coset oracle and no entangled multi-copy measurement.

\section{Scope and extensions}
\label{sec:scope}

The weak demonstration is the cyclic case and the strong demonstration the two-system
case; the pipeline is stated for any finite group whose symmetry-adapted basis is
available, and the Schur and non-Abelian instances are extensions rather than demonstrated
results. For non-Abelian groups the matched basis block-diagonalizes rather than
diagonalizes, and the multiplicity within a block sets the within-block reduction.

The discovery cost is governed by the score gap, which is model-dependent. For a target
with little structure the gap between a match and an over-selection is small and discovery
is expensive, but that is precisely the regime where the exploitation reduction is also
small, so the pipeline does not spend discovery shots where it could not win the
exploitation back. The strong mode adds a restriction of its own: confinement to a single
sector is special, and the strong score certifies it only when the target genuinely lives
there; the weak mode applies whenever the target merely commutes with the representation.

The candidate family is fixed in advance and ordered by reduction. Discovery over an
unstructured family, the non-Abelian instance of the commutant scoring, and a
hardware-level accounting of the basis-change circuit are the natural next steps.

\section{Conclusion}

When the symmetry a quantum target carries is not known in advance, it can be discovered
from the data by scoring a family of candidate groups with a symmetry test, and then
exploited as the measurement basis to read the target in fewer shots. The same normalized
score, with the weak twirl or the strong sector projection, detects either notion of
symmetry, and the selection rule that takes the largest passing group is unbiased. Charged
in full for its own discovery, the loop wins: weak matching on momentum readout reduces
shots by a factor that widens from about ten to several thousand, and strong matching on a
two-system target reduces them by a further factor of the subsystem size. Blind symmetry
matching is a practical shot-reduction preprocessor for the common case where the matched
basis cannot be written down in advance.

\bibliographystyle{IEEEtran}
\bibliography{refs}
\end{document}